\documentclass[10, conference]{IEEEtran}

\usepackage{amsmath, amsfonts, amsthm,  graphicx}
\usepackage{cite}

\newtheorem{theorem}{Theorem}
\newtheorem{lemma}{Lemma}

\newcounter{example}
\newenvironment{example}[1][]{\refstepcounter{example}\par\medskip\noindent%
   \textbf{Example~\theexample. }}{\medskip}

\newcommand{\set}[1]{\mathcal{#1}}

\long\def\symbolfootnote[#1]#2{\begingroup
\def\thefootnote{\fnsymbol{footnote}}\footnote[#1]{#2}\endgroup}

\title{Construction of Short Protocol Sequences with Worst-Case Throughput Guarantee}

\author{Kenneth W. Shum and Wing Shing Wong \\ Department of Information Engineering \\ The Chinese University of Hong Kong. \\
{\small \tt wkshum@inc.cuhk.edu.hk}, {\small \tt wswong@ie.cuhk.edu.hk}}

\begin{document}

\maketitle

\begin{abstract}
Protocol sequences are used in channel access for the multiple-access collision channel without feedback. A new construction of protocol sequences with a guarantee of worst-case system throughput is proposed. The construction is based on Chinese remainder theorem. The Hamming crosscorrelation is proved to be concentrated around the mean. The sequence period is much shorter than existing protocol sequences with the same throughput performance. The new construction reduces the complexity in implementation and also shortens the waiting time until a packet can be sent successfully.
\end{abstract}

{\it Tags:} Protocol sequences, collision channel without feedback, wobbling sequences.

\section{Introduction}
\symbolfootnote[0]{This work was partially supported by a grant from the Research Grants Council of the Hong Kong Special Administrative Region under Project 417909.}
Protocol sequences are periodic binary sequences for multiple-access control in the collision channel without feedback~\cite{Massey85}. In a time-slotted scenario, each user repeatedly reads out the value of a statically assigned protocol sequence, and sends a packet in a time slot if the sequence value is equal to one. If two or more users transmit simultaneously in the same time slot, there is a collision and the collided packets are assumed unrecoverable. If there is exactly one  transmitting user while the others remain silent, the received packet is assumed error-free. We do not assume any feedback from the receiver and any coordination among the transmitters. This assumption is applicable to low-cost and low-complexity wireless sensor networks, as it is not necessary to spare any hardware on monitoring the channel; the transmitters simply send a packet whenever the value of the assigned protocol sequence is one, regardless of the channel condition. For simplicity in presentation, we assume slot-synchronization in this paper. This requirement can be relaxed without much degradation in performance.

This channel model is considered in practical sensor networks, such as f-MAC~\cite{Roedig}. One of the main design issues is the construction of protocol sequences. Some design objectives are addressed in~\cite{SWSC09}. Since we do not assume any coordination among the users, the users may not start their protocol sequences at the same time. This incurs relative delay offsets among the transmitters. Our first objective is to design protocol sequences such that no matter what the relative delay offsets are, the system throughput is provably larger than some positive constant. This provides a  throughput guarantee in the worst-case.

The second objective is to minimize the sequence period, which measures the delay one has to wait until the promised number of successful packets go through the channel. For protocol sequences with very long period, a user may suffer starvation in the short term even though the throughput in the whole period is very good. This issue is alleviated if the sequence period is minimized.

The two objectives mentioned above are contradicting, and there is a tradeoff between them. In~\cite{Massey85, Rocha00, SWSC09}, a class of protocol sequences, called {\em shift-invariant} sequences are studied. This class of protocol sequences can achieve optimal system throughput, but the sequence period grows exponentially as a function of the users, and hence is not of  practical interests when the number of users is large. Another class of protocol sequences, called {\em wobbling sequences} is constructed in~\cite{Wong07}. The system throughput is provably larger than 0.25 for any choice of relative delay offsets, and the sequence period grows like $M^4$, where $M$ is for the number of users in the system. In this paper, we construct shorter protocol sequences with roughly the same throughput performance by the wobbling sequences.

Other constructions of protocol sequences are investigated in~\cite{GyorfiVajda93,  MZKZ95}, sometime under the name of {\em cyclically permutable constant weight codes} (CPCWC). The difference between CPCWC and our protocol sequences is that the latter only has Hamming crosscorrelation requirement, but the former has both autocorrelation and crosscorrelation constraints.
In~\cite{Chu06}, constructions using optical orthogonal codes and cyclic superimposed codes are considered.

In Section~\ref{sec:CRT}, the new construction of protocol sequences is described. In Section~\ref{sec:correlation}, we investigate the crosscorrelation properties, which are crucial in the derivation of a lower bound on system throughput in Section~\ref{sec:bound}. Comparison with shift-invariant and wobbling sequences is given in Section~\ref{sec:bound}.

\section{Construction of Protocol Sequences via Chinese Remainder Theorem} \label{sec:CRT}

We will use sequence ``period'' and sequence ``length'' interchangeably.
Let $\mathbb{Z}_n = \{0,1,\ldots, n-1\}$ be the residues of integers mod $n$. The components in a sequence of length $L$ are indexed by $\mathbb{Z}_L$. The {\em Hamming weight} of a binary sequence $a(t)$ of length $L$ is the number of ones in $a(t)$ in a period. For two binary sequences $a(t)$ and $b(t)$ of length $L$, their {\em Hamming crosscorrelation function} is defined as
\[
 H_{ab}(\tau) := \sum_{t\in\mathbb{Z}_L} a(t) b(t+\tau),
\]
where $\tau$ is the delay offset.

We shall construct sequences with length $L = pq$, where $p$ and $q$ are two relatively prime integers. In this paper, we will take $p$ to be a prime number and $q$ an positive integer not divisible by~$p$.
By the Chinese Remainder Theorem (CRT)~\cite{IrelandRosen}, there is a bijection between $\mathbb{Z}_{pq}$ and
the direct sum
\[
G_{p,q} := \mathbb{Z}_p \oplus \mathbb{Z}_q.
\]
The bijective mapping  $\Phi: \mathbb{Z}_{pq} \rightarrow G_{p,q}$ is given by
$$
 \Phi(x) = (x \bmod p, x \bmod q).
$$
Henceforth, we will identify $\mathbb{Z}_{pq}$ with $G_{p,q}$.

{\bf CRT Construction} Given a prime number $p$ and an integer $q$ relatively prime to $p$, we define a sequence $s_g(t)$ of length $pq$, for $g=0,1,2\ldots, p-1$, by
\[
 s_g(t) := \begin{cases}
 1 & \Phi(t) = (\bar{j}g, j) \text{ for some }j,\ 0\leq j < q, \\
 0 & \text{otherwise},
 \end{cases}
\]
where $\bar{j}$ is the residue of $j$ in $\mathbb{Z}_p$.
We call the sequence $s_g(t)$ the {\em CRT sequence generated by~$g$}. The integer $g$ is called the {\em generator} of $s_g(t)$.

Alternately, we can define the CRT sequences by specifying their characteristic sets. For each $g \in \{0,1,\ldots,p-1\}$,  let
\begin{equation}
\set{I}_{g}:=  \{ (\bar{j}g ,j) \in G_{p,q}:\, 0\leq j < q\}.
\label{eq:Ig}
\end{equation}
We note that
$\set{I}_{g}$ is an arithmetic progression in $G_{p,q}$ with common difference $(g,1)$. The CRT sequence with generator $g$ is obtained by setting $s_g(t) = 1$ if and only if $\Phi(t) \in \set{I}_g$.

The Hamming weight of each sequence is equal to~$q$.

\begin{example} $p=5$ and $q = 9$.
The five CRT sequences $s_0(t)$ to $s_4(t)$ are listed as follows. The common period is 45, and the Hamming weight of each sequence is equal to~9.
{\small
\begin{gather*}
s_0:\ 10000\,10000\,10000\,10000\,10000\,10000\,10000\,10000\,10000,\\
s_1:\ 11111\,11110\,00000\,00000\,00000\,00000\,00000\,00000\,00000,\\
s_2:\ 10000\,10000\,00010\,00000\,01000\,01001\,00001\,00100\,00100,\\
s_3:\ 1 0 0 0 0\, 1 0 0 0 0\, 0 1 0 0 0\, 0 1 0 0 0\, 0 0 1 0 0\, 0 0 0 1 0\, 0 0 0 1 0\, 0 0 0 0 1\, 0 0 0 0 1, \\
s_4:\ 1 0 0 0 0\, 1 0 0 0 0\, 0 0 1 0 0\, 0 0 1 0 1\, 0 0 0 0 1\, 0 0 0 0 0\, 0 1 0 0 0\, 0 0 0 1 0\, 0 0 0 1 0.
\end{gather*} }
\label{ex:A}
\end{example}

\vspace{-0.5cm}

\section{Crosscorrelation Properties} \label{sec:correlation}

The main idea of using CRT in constructing protocol sequences hinges on the fact that $\Phi$ is a homomorphism of abelian groups,
so that the analysis of crosscorrelation can be carried out in $G_{p,q}$ instead of $\mathbb{Z}_{pq}$.
We remark that in~\cite{Nguyen92, GyorfiVajda93, MZKZ95, SWSC09}, the idea of CRT appears in the same way as in this paper.

Given a one-dimensional delay $\tau$, we denote its two-dimensional counterpart by $(\tau_1, \tau_2) := \Phi(\tau)$. For $h \in\mathbb{Z}_p$, define the translation $\set{I}_h+(\tau_1, \tau_2)$ of $\set{I}_h$ by $(\tau_1,\tau_2)$ by
\[
\{(x+\tau_1,y+\tau_2)\in G_{p,q} :\, (x,y) \in \set{I}_h \}.
\]
The Hamming crosscorrelation $H_{gh}(\tau)$ can be computed by
\[
H_{gh}(\tau) = |I_g \cap (I_h-(\tau_1, \tau_2))|,
\]
where $|\mathcal{A}|$ indicates the cardinality of a set~$\mathcal{A}$. To distinguish arithmetic operations in $\mathbb{Z}_p$ and $\mathbb{Z}_q$, we use $\oplus_p$ and $\ominus_p$ for addition and subtraction in $\mathbb{Z}_p$, and $\oplus_q$ and $\ominus_q$ for addition and subtraction in $\mathbb{Z}_q$. In this notation, we have
\[
 \set{I}_h +(\tau_1,\tau_2) = \{(\bar{j} h \oplus_p \tau_1, j \oplus_q \tau_2):\, j=0,1,\ldots, q-1  \}.
\]
By a change of variable, $ \set{I}_h +(\tau_1,\tau_2) $ can be written as
\[
 \{ ((( \overline{j \ominus_q \tau_2})h) \oplus_p \tau_1, j):\, j=0,1,\ldots, q-1\}.
\]
After comparing with the definition of $\set{I}_g$ in~\eqref{eq:Ig}, we see that $|\set{I}_g \cap (\set{I}_h + (\tau_1, \tau_2))|$ is equal to the number of solutions to
\begin{equation}
 \bar{x}g  \equiv ((\overline{x \ominus_q \tau_2})h) \oplus_p  \tau_1 \bmod p. \label{eq:Ham}
\end{equation}
for $x = 0,1,\ldots, q-1$. The problem of computing the crosscorrelation function is thus reduced to counting the solutions to~\eqref{eq:Ham}.

The following simple lemma is useful in the derivation of Hamming crosscorrelation.

\begin{lemma} Let $p$ be a prime number. For each $b \in\mathbb{Z}_p$, the number of solutions to
$
\bar{x} \equiv b \bmod p
$
for $x$ going through $d$ consecutive integers $c, c+1, \ldots c+d-1$, equals
\[
  \begin{cases}
   d/p  & \text{if $p$ divides $d$}, \\
  \lfloor d/p \rfloor + \delta & \text{otherwise,}
   \end{cases}
\]
where $\delta$ is either 0 or~1. \label{lemma:simple}
\end{lemma}

\begin{proof}
In the first case where $d$ is divisible by~$p$, if we reduce the integers $c$, $c+1, \ldots, c+d-1$ mod $p$, we have each element in $\mathbb{Z}_p$ repeated exactly $d/p$ times. Hence, for each $b \in \mathbb{Z}_p$, there are exactly $d/p$ integers in $\{c, \ldots, c+d-1\}$ whose residue mod $p$ equal~$b$

For the second case, where $d$ is not divisible by~$p$, we divide the $d$ consecutive integers into two parts. The first part consists of $\lfloor d/p \rfloor p$ consecutive integers and the second part consists of the last $d-\lfloor d/p \rfloor p$ integers. Among the first $\lfloor d/p \rfloor p$ integers, exactly $\lfloor d/p \rfloor$ of them equal $b$ mod $p$. The residues of the $d - \lfloor d/p \rfloor p$ integers in the second part are distinct, and hence at most one of them is equal to $b$. The number of integers in $\{c, \ldots, c+d-1\}$ whose residue equal $b$ is either $\lfloor d/p \rfloor$ or $\lfloor d/p \rfloor + 1$.
\end{proof}

From Lemma~\ref{lemma:simple},
we obtain the crosscorrelation between $s_0(t)$ and other sequences.

\begin{theorem}
For $g\neq 0$, the Hamming crosscorrelation of $s_g(t)$ and $s_0(t)$ is equal to either $\lfloor q/p \rfloor$ or $\lfloor q/p \rfloor +1$. \label{thm:cross-correlation0}
\end{theorem}

\begin{proof}
If we put  $h=0$ in~\eqref{eq:Ham}, we get
$  \bar{x} \equiv g^{-1} \tau_1 \bmod p$.
The number of integers in $\{0,1,\ldots, q-1\}$ which equal $g^{-1} \tau_1 \bmod p$ is either $\lfloor q/p \rfloor$ or $\lfloor q/p \rfloor +1$
by Lemma~\ref{lemma:simple}.
\end{proof}

For nonzero $h$, we can divide both sides of \eqref{eq:Ham} by $h$ and re-write it as
\[
  \bar{x} (h^{-1} g) \equiv (\overline{x \ominus_q \tau_2}) \oplus_p (h^{-1} \tau_1) \bmod p.
\]
For each fixed $\tau_2$, as $\tau_1$ runs through $\mathbb{Z}_p$, $h^{-1} \tau_1$ also runs through the complete set of residues mod~$p$. Therefore, the distribution of Hamming crosscorrelation between $s_g(t)$ and $s_h(t)$ is the same as the distribution of Hamming cross-correlation between $s_{g/h}(t)$ and $s_1(t)$. We henceforth focus on the case $h=1$ without any loss of generality.

To aid the derivation of the Hamming crosscorrelation, we first prove the following lemma.

\begin{lemma} Let $g \in \mathbb{Z}_p \setminus \{1\}$, and denote $\Phi(\tau)$ by $(\tau_1, \tau_2)$. The Hamming crosscorrelation between $s_g(t)$ and $s_1(t)$, namely $H_{g1}(\tau)$, satisfies the following properties:

\begin{enumerate}
\item $H_{g1}(\tau)$ equals the number of solutions to
\begin{equation}
\bar{x} \equiv a_g(\tau_1,\tau_2) + b_g \mathbf{I}(0\leq x < \tau_2) \bmod p,
\label{eq:Ham4}
\end{equation}
for $x= 0,1,\ldots, q-1$, where
\begin{align}
a_g(\tau_1,\tau_2) &:= (g-1)^{-1}(\tau_1 - \bar{\tau}_2), \\
 b_g &:= (g-1)^{-1} \bar{q},  \label{eq:def_b}
\end{align}
and $\mathbf{I}$ is the indicator function defined as
\[
\mathbf{I}(P) := \begin{cases}
1 & \text{if $P$ is true}, \\
0 & \text{if $P$ is false}.
\end{cases}
\]
\label{lemma:H2}

\item Let $\tau$ and $\tau'$ denote two relative delay offsets. Suppose that the first component of $\Phi(\tau)$ and $\Phi(\tau')$ are the same, and the second component of $\Phi(\tau)$ and $\Phi(\tau')$ defer by a multiple of $p$, then then
  $H_{g1}(\tau) = H_{g1}(\tau')$.

\label{lemma:H3}
\end{enumerate}
    \label{lemma:H}
\end{lemma}

\begin{proof}
After setting $h$ in \eqref{eq:Ham} to 1, we obtain
\begin{equation}
 \bar{x} g \equiv  (\overline{x \ominus_q \tau_2}) \oplus_p \tau_1 \bmod p.
 \label{eq:Ham_app}
\end{equation}
We want to show that the number of solutions to~\eqref{eq:Ham_app}, for $q=0,1,\ldots, q-1$, is the same as the number of solutions to~\eqref{eq:Ham4}.

We consider $x$ in two disjoint ranges: (i) $0 \leq x <\tau_2$, and (ii) $\tau_2 \leq x < q$. In the first case, $x \ominus_q \tau_2$ is congruent to $x+q-\tau_2 \bmod q$. So, for $0\leq x <  \tau_2$, \eqref{eq:Ham_app} is equivalent to
\begin{equation}
 \bar{x}g \equiv \bar{x}+\bar{q}-\bar{\tau}_2+\tau_1 \bmod p
\label{eq:Ham2}
\end{equation}
where  $\bar{q}$ and $\bar{\tau}_2$ are residues of  $q$ and $\tau_2$ in $\mathbb{Z}_p$, respectively.

In the second case, for $x=\tau_2, \tau_2+1, \ldots, q-1$, \eqref{eq:Ham_app} is equivalent to
\begin{equation}
\bar{x}g \equiv \bar{x}-\bar{\tau}_2+\tau_1 \bmod p.
\label{eq:Ham3}
\end{equation}

We combine \eqref{eq:Ham2} and \eqref{eq:Ham3} in one line as
\[
\bar{x}(g-1) \equiv  -\bar{\tau}_2+\tau_1 + \bar{q} \mathbf{I}(0\leq x < \tau_2) \bmod p,
\]
Since $g$ is not equal to 1 by assumption, we can divide by $(g-1)$ and obtain~\eqref{eq:Ham4}. This proves the first part of the lemma.

The second part of the lemma is vacuous if $q<p$. So we assume $q>p$. (The case $q=p$ is excluded because it is assumed that $q$ is relatively prime with $p$.) Let $(\tau_1,\tau_2) = \Phi(\tau)$ and $(\tau_1',\tau_2') = \Phi(\tau')$.
It is sufficient to prove the statement for $\tau_1 = \tau_1'$ and $\tau_2' = \tau_2+p$, namely, the number of solutions to~\eqref{eq:Ham4} and the number of solutions to
\begin{equation}
\bar{x} \equiv a_g(\tau_1, \tau_2') + b_g \mathbf{I}(0\leq x <\tau_2') \bmod p \label{eq:LemmaB}
\end{equation}
for $x=0,1,\ldots, q-1$, are the same. We note that $a_g(\tau_1, \tau_2)$ is equal to $a_g(\tau_1, \tau_2')$.  However, the arguments inside the indicator function are different. We divide the range of $x$ into three disjoint parts:
\begin{align*}
\set{X}_1 &:= \{0,1,\ldots, \tau_2-1\}, \\
\set{X}_2 &:= \{\tau_2,\tau_2+1,\ldots, \tau_2+p-1\}, \\
\set{X}_3 &:= \{\tau_2+p,\tau_2+p+1,\ldots, q-1\}.
\end{align*}
Since $q>p$, $\set{X}_3$ is non-empty.
For $x \in \set{X}_1$, $\mathbf{I}(0\leq x < \tau_2) = \mathbf{I}(0\leq x < \tau_2')$.
Therefore
\eqref{eq:Ham4} and~\eqref{eq:LemmaB} have the same number of solutions for $x$ in $\set{X}_1$. For $x\in \set{X}_2$, both \eqref{eq:Ham4} and \eqref{eq:LemmaB}
have exactly one solution by Lemma~\ref{lemma:simple}.
For $x \in \set{X}_3$, \eqref{eq:Ham4} is equivalent to \eqref{eq:LemmaB}, and hence has the same number of solutions as~\eqref{eq:LemmaB} does. In conclusion, the number of solutions to \eqref{eq:Ham4} and \eqref{eq:LemmaB} for $x \in\mathcal{X}_1\cup\mathcal{X}_2\cup\mathcal{X}_3$ are the same. This finishes the proof of the second part of the lemma.
\end{proof}

From now on, we assume that $q>p$, which is the case of practical interest.

\begin{theorem} Let $p$ and $q$ be positive integers such that $p$ is prime, $\gcd(p,q)=1$ and $q>p$. Let $m$ denote the quotient of $q$ divided by~$p$, i.e., $m=\lfloor q/p \rfloor$, and let $g\in\mathbb{Z}_p$, $0 \neq g\neq 1$. Let $\bar{q}$ be the residue of $q$ mod $p$, and $b_g$ be defined as in~\eqref{eq:def_b}.
The Hamming crosscorrelation between $s_g(t)$ and $s_1(t)$ is bounded between
\begin{align}
m-1 \text{ and } m+1 & \text{    if } 0 < b_g < p-\bar{q}, \text{ or } \label{thm:case1} \\
m \text{ and } m+2 & \text{    if } p-\bar{q} < b_g < p. \label{thm:case2}
\end{align}
\label{thm:cross-correlation1}
\end{theorem}

\begin{proof}
By the second part of the previous lemma, we only need to consider $\tau_2 = 0,1,\ldots, p-1$. In this proof, we will denote $H_{g1}(\tau)$ by $H_{g1}(\tau_1, \tau_2)$, with $(\tau_1, \tau_2)$ equal to $\Phi(\tau)$.

We first prove the first case in~\eqref{thm:case1}  by considering two cases.

\underline{Case 1,  $0\leq \tau_2 < \bar{q}$}

Suppose  that~\eqref{eq:Ham4} has no solution for $0\leq x< \tau_2$. As the indicator function in~\eqref{eq:Ham4} is zero for $x = \tau_2$, $\tau_2+1, \ldots, q-1$, \eqref{eq:Ham4} is reduced to
\[
 \bar{x} \equiv a_g(\tau_1, \tau_2) \bmod p.
\]
The number of integers in $\{ \tau_2$, $\tau_2+1, \ldots, q-1\}$, say $d$, satisfies $\lfloor d/p \rfloor = m$.
By Lemma~\ref{lemma:simple}, we have either $m$ or $m+1$ solutions to~\eqref{eq:Ham4}  for $x \geq \tau_2$.

Secondly, suppose that~\eqref{eq:Ham4} has exactly one solution for $0\leq x< \tau_2$. The indicator function in~\eqref{eq:Ham4} is equal to~1 for $0\leq x < \tau_2$. Hence,
\begin{equation}0\leq a_g(\tau_1,\tau_2)+b_g < \tau_2. \label{eq:contradiction1}
 \end{equation}
We claim that \eqref{eq:Ham4} has no solution for $x=\tau_2, \tau_2+1, \ldots, \bar{q}-1$. Otherwise, we have
\[
\tau_2 \leq a_g(\tau_1,\tau_2) < \bar{q},
\]
which, after combining with the assumption that $1 \leq b_g \leq p - \bar{q}-1$, yields
\[
\tau_2 < a_g(\tau_1,\tau_2) + b_g < p-1.
\]
This contradicts with~\eqref{eq:contradiction1} and proves the claim. For $$\bar{q}\leq x < q,$$ there are exactly $m$ solutions by Lemma~\ref{lemma:simple}. The total number of solutions to~\eqref{eq:Ham4} for $x=0,1,\ldots, q-1$, is thus $m+1$. Hence $H_{g1}(\tau_1, \tau_2) = m+1$.

\underline{Case 2:  $\bar{q} \leq \tau_2 < p$}

By Lemma~\ref{lemma:simple}, \eqref{eq:Ham4} has either 0 or 1 solution for $0\leq x \leq \tau_2$, and either $m-1$ or $m$ solutions for $\tau_2 \leq x < q$. Hence, $H_{g1}(\tau_1, \tau_2)$ is within the range of $\{m-1, m, m+1\}$.

\medskip

For $b_g = p-\bar{q}+1, \ldots, p-1$, we again consider two cases.

\underline{Case 1: $0 \leq \tau_2 < \bar{q}$}

By Lemma~\ref{lemma:simple}, \eqref{eq:Ham4} has either 0 or 1 solution for $0\leq x < \tau_2$, and either $m$ or $m+1$ solutions for $\tau_2 \leq x < q$. Therefore, $H_{g1}(\tau_1, \tau_2) \in \{m, m+1, m+2\}$.

\underline{\em Case 2:  $\bar{q} \leq \tau_2 < p$}

Suppose that \eqref{eq:Ham4} has no solution for $0 \leq x < \tau_2$, i.e.,
\begin{equation}
\tau_2 \leq a_g(\tau_1, \tau_2) + b_g < p. \label{eq:contradiction2}
\end{equation}
We claim that \eqref{eq:Ham4} must have one solution for $x$ in the following range
\begin{equation}
\tau_2 \leq x < p+\bar{q}. \label{eq:range}
\end{equation}
From the assumption of $\bar{q} \leq \tau_2 < p$, we deduce that
$$\bar{q}<p+\bar{q}-\tau_2 \leq p,$$ so that the range in~\eqref{eq:range} is non-empty and consists of no more than $p$ integers.
If the claim were false, we would have no solution to~\eqref{eq:Ham4} for $\tau_2 \leq x < p+\bar{q}$, implying that
\begin{equation}
\bar{q} \leq a_g(\tau_1, \tau_2) < \tau_2. \label{eq:appA}
\end{equation}
Here, we have used the fact that the indicator function in~\eqref{eq:Ham4} is equal to zero for $x$ in the range in~\eqref{eq:range}.
By adding~\eqref{eq:appA} to $$p-\bar{q}+1\leq b_g \leq p-1$$ and reducing mod $p$, we obtain
\[
1 \leq  a_g(\tau_1, \tau_2) + b_g < \tau_2,
\]
which is a contradiction to~\eqref{eq:contradiction2}. Thus, the claim is proved. For $x = p+\bar{q}, p+\bar{q}+1, \ldots, q-1$, there are exactly $m-1$ solutions to~\eqref{eq:Ham4} by Lemma~\ref{lemma:simple}. Totally there are $m$ solutions, and thus $H_{g1}(\tau_1,\tau_2) = m$.

Finally suppose that \eqref{eq:Ham4} has exactly one solution for $0 \leq x < \tau_2$. As the number of solutions to~\eqref{eq:Ham4} for $x=\tau_2, \tau_2+1, \ldots, q-1$ is either $m-1$ or $m$ by Lemma~\ref{lemma:simple}, the total number of solutions to~\eqref{eq:Ham4} is either $m$ or $m+1$.

In any case, we see that $H_{g1}(\tau_1, \tau_2)$ is either $m$, $m+1$ or $m+2$.
\end{proof}

Thereom~\ref{thm:cross-correlation1} asserts that for any pair of distinct CRT sequences, the Hamming crosscorrelation is either between $m-1$ and $m+1$, or between $m$ and $m+2$. For the whole sequence set, the Hamming crosscorrelation is therefore four-valued. We next show that for some special choice of $q$, namely $q \equiv \pm 1 \bmod p$, the Hamming crosscorrelation of the whole sequence set assumes only three distinct values.

\begin{theorem} Let $p$ and $q$ be integers as in Theorem~\ref{thm:cross-correlation1}.

\begin{enumerate}

\item If $q$ is of the form $mp+1$ for some positive integer $m$, then for $g=2,3,\ldots, p-1$, $H_{g1}(\tau)$ is between $m-1$ and $m+1$.

\item If $q$ be of the form $mp+(p-1)$ for some positive integer $m$, then for $g = 2,3,\ldots, p-1$, $H_{g1}(\tau)$ is between $m$ and $m+2$.
\end{enumerate}
\label{thm:cross-correlation2}
\end{theorem}

\begin{proof}
For he first part of the theorem, we have $\bar{q}$ equal to 1 mod~$p$.  So
\[
b_g \equiv (g-1)^{-1} \bar{q} \equiv (g-1)^{-1} \bmod p.
\]
Since $g$ is  between $2$ and $p-1$ inclusively, $g-1$ is between $1$ and $p-2$, and hence the inverse of $g-1$ mod $p$ is also between $1$ and $p-2$. We thus obtain
$ 0 < b_g < p-1 = p-\bar{q}$.
The result now follows from Theorem~\ref{thm:cross-correlation1}.

The second part can be proved similarly from Theorem~\ref{thm:cross-correlation1} by putting $\bar{q} = p-1$.
\end{proof}

Together with Theorem~\ref{thm:cross-correlation0}, which says that the Hamming crosscorrelation between $s_0(t)$ and $s_g(t)$, for $g\neq 0$, is either $m$ or $m+1$, we prove that the Hamming crosscorrelation of the whole CRT sequence set is three-valued when $q \equiv \pm 1 \bmod p$. The sequences in Example~1 are generated with $q \equiv -1 \bmod p$. We can verify that the Hamming crosscorrelation in Example~1 is either 1, 2 or~3.

\section{Lower Bound on System Throughput} \label{sec:bound}

The three-valued result in Theorem~\ref{thm:cross-correlation2} suggests that the variation of Hamming crosscorrelation due to relative delay offsets is minimal when $q \equiv \pm 1 \bmod p$. We single out the $q \equiv -1 \bmod p$ case below and derive a lower bound on the resulting system throughput. The  case of $q \equiv 1 \bmod p$ is similar and omitted.

When $q$ is of the form  $kp-1$ for some integer $k \geq 2$, the CRT construction  yields $p$ protocol sequences of length $L=kp^2-p$ and Hamming weight $kp-1$. By Theorem~\ref{thm:cross-correlation2}, the largest Hamming crosscorrelation value is $k+1$.

We pick $M$ sequences and form the CRT sequence set of size~$p$ in order to support $M$ users.
Here, $M$ is an integer  whose value will be optimized later. Since a user sends $kp-1$ packets in a period, and each other user may collide with him in at most $k+1$ packets, the number of successful packets per user per period is no less than $kp-1 -(M-1)(k+1)$. The total number of successful packets, summed over all $M$ users, is thus lower bounded by
\begin{equation}
M  [kp-1 - (M-1)(k+1)].  \label{eq:lower_bound}
\end{equation}
By the method of completing square, we can write~\eqref{eq:lower_bound} as
\begin{equation}
 (k+1) \left[ -\Big(M-\frac{k(p+1)}{2(k+1)}\Big)^2 + \Big(\frac{k(p+1)}{2(k+1)} \Big)^2 \right].
\end{equation}
We see that the maximum value in~\eqref{eq:lower_bound} is obtained when
\begin{equation}
 M^* = k(p+1)/(2(k+1)). \label{eq:best_M}
\end{equation}

Since $M$ must be an integer, after taking the floor of~\eqref{eq:best_M}, we obtain the following theorem.

\begin{theorem}
Let $p$ be a prime number, $k\geq 2$, and $M'$ be the largest integer smaller than or equal to $M^*$ in~\eqref{eq:best_M}.
By picking $M'$ sequences from the CRT construction with parameters $p$ and $q=kp-1$, the system throughput is lower bounded by
\begin{equation}
\frac{1}{p(kp-1)} \left[ \frac{(p+1)^2}{4} \cdot \frac{k^2}{k+1} - (k+1)\right] \label{eq:throughput}
\end{equation}
\label{thm:throughput}
\end{theorem}

\begin{proof}
Consider the expression in~\eqref{eq:lower_bound} as a function of $M$, and denote it by $f(M)$.
Since the difference between $M^*$ and $M'$ is at most one,
$f(M^*) - f(M') \leq (k+1)(M' - M^*) \leq k+1$.
After division by the period $p(kp-1)$, we have the following lower bound on system throughput,
$$
\frac{f(M')}{p(kp-1)} \geq \frac{f(M^*) - (k+1)}{p(kp-1)}
$$
which can be readily seen to be the same as~\eqref{eq:throughput}.
\end{proof}

We note that the value in~\eqref{eq:throughput} is approximately equal to 0.25 when $k$ and $p$ are large.

Theorem~\ref{thm:throughput} provides a hard guarantee on the worst-case system throughput; no matter what the delay offsets are, the system throughput is always larger than the value in~\eqref{eq:throughput}. Theorem~\ref{thm:throughput} also indicates a tradeoff between the the lower bound and the sequence period. If we increase the value of $k$, the sequence period is increased, but the lower bound on system throughput is also increased.

We remark that the lower bound in Theorem~\ref{thm:throughput} is not tight. The actual system throughput is higher than~\eqref{eq:throughput}. The next example compares the lower bound with the average throughput over relative delay offsets.

\begin{example} We consider an example with $M=19$ users, using CRT sequences with $p=37$ and $q=kp-1$. The throughput is plotted against the sequence period, while keeping the fraction of ones in each sequence fixed at $1/p$. This means that the fraction of time in which a user is transmitting, and hence the power of each user, is kept constant. We compare the lower bound in~\eqref{eq:throughput} with the average throughput obtained by simulation in Fig.~\ref{fig:simulate}. For each~$k$, 20000 delay offset combinations are randomly generated. The mean system throughput is about 0.31
In addition to the mean throughput, the maximum and minimum throughput obtained among these 20000 delay offset combinations are also plotted. The variation of throughput diminishes as sequence period increases. The value of the lower bound~\eqref{eq:throughput} also increases.
We see that the minimal observed system throughput is much higher than the lower bound in Theorem~\ref{thm:throughput}.

For the shift-invariant protocol sequence set for nineteen users, the sequence period is $19^{19}$, which is astronomical. Nevertheless, it has higher system throughput $e^{-1}=0.3679$. For wobbling sequences, a lower bound of 0.25 system throughput~\cite[(6.12)]{Wong07} is guaranteed when the sequence period is $19^4 \approx 1.3\times 10^5$. From Fig.~\ref{fig:simulate}, a lower bound of 0.25 can be obtained by using CRT sequences when the sequence period is about $1.1\times 10^4$, a roughly ten-fold reduction in sequence period.
\end{example}

\begin{figure}
\begin{center}
  \includegraphics[width=3.5in]{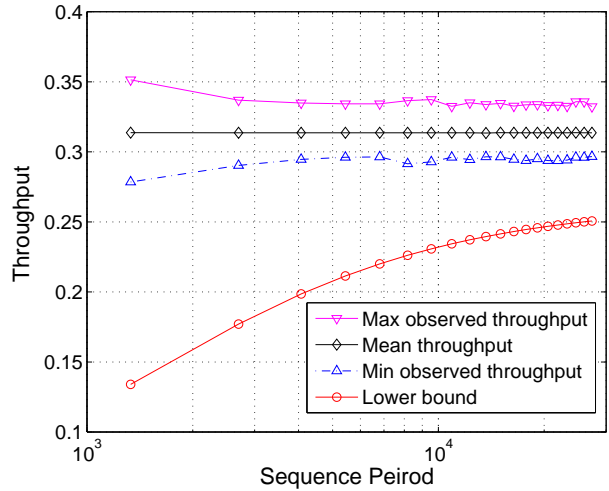}
\end{center}
\caption{System Throughput of CRT Sequences for $M=19$ users, $p=37$.}
\label{fig:simulate}
\end{figure}

\section{Conclusion}

In order to minimize the waiting time until a successful packet is sent, while maintaining a high level of system throughput,
a class of  protocol sequences with short period are constructed. After analyzing the crosscorrelation properties, we derive a lower bound on the system throughput. The constructed sequences provides flexibility on the tradeoff between sequence period and the worst-case system throughput.




\end{document}